\documentclass{article}

\usepackage{amsfonts}
\usepackage{amssymb}
\usepackage{amsmath}
\usepackage{amstext}
\usepackage{amscd}
\usepackage{amsbsy}
\usepackage{amssymb}

\newtheorem{theorem}{Theorem}[section]
\newtheorem{proposition}[theorem]{Proposition}

\newtheorem{corollary}[theorem]{Corollary}
\newtheorem{remark}{Remark}[section]
\newtheorem{example}{Example}[section]

\newenvironment{proof}[1][\emph{Proof.}]{\noindent {#1} }{\hfill \rule{0.5em}{0.5em}\\}

\numberwithin{equation}{section}

\title{Homological Equations for Tensor Fields and Periodic  Averaging}

\author{Misael Avenda\~{n}o Camacho \quad and \quad Yury M. Vorobiev}
\date{\small Department of Mathematics, University of Sonora\\
Rosales y Blvd. Luis Encinas, Hermosillo, M\'exico, 83000.}

\begin{document}

\maketitle

\begin{abstract}
Homological equations of tensor type associated to periodic flows on
a manifold are studied. The Cushman intrinsic formula \cite{Cush-84}
is generalized to the case of multivector fields and differential
forms. Some applications to normal forms and the averaging method
for perturbed Hamiltonian systems on slow-fast phase spaces are
given.
\end{abstract}

\section{Introduction}
The so-called homological equations usually appear in the context of
normal forms and the method of averaging for perturbed dynamical
systems (see, for example, \cite{ArKN-88,SanVer-85}). According to
the Lie transform method \cite{Dep-69,Hen-70}, the infinitesimal
generators of normalization transformations for perturbed dynamics
systems are defined as the solutions to homological equations for
vector fields. In the Hamiltonian case, the normalization problem is
reduced to the solvability of homological equations for functions.
Let $X$ be a vector field on a manifold $M$ whose flow is periodic
with period function $T:M\rightarrow\mathbb{R}.$ Then, it is
well-known \cite{Cush-84} that for a given $G\in C^{\infty}(M)$
there exist smooth functions $F$ and $\bar{G}$ on $M$ satisfying the
homological equation
\begin{equation}
\mathcal{L}_{X}F=G-\bar{G} \label{Int1}
\end{equation}
and the condition
\begin{equation}
\mathcal{L}_{X}\bar{G}=0. \label{Int2}
\end{equation}
The solvability of this problem follows from the decomposition
\begin{equation}
C^{\infty}(M)=\operatorname{Ker}\mathcal{L}_{X}\oplus\operatorname{Im}
\mathcal{L}_{X} \label{Int3}
\end{equation}
and the corresponding global solutions are given by the formulas
\cite{Cush-84}
\begin{equation}
\bar{G}=\frac{1}{T}\int_{0}^{T}G\circ\operatorname{Fl}_{X}^{t}dt, \label{Int4}
\end{equation}
\begin{equation}
F=\frac{1}{T}\int_{0}^{T}(t-\frac{T}{2})G\circ\operatorname{Fl}_{X}^{t}dt+K,
\label{Int5}
\end{equation}
where $K$ is a first integral of $X$. This result together with
Deprit's algorithm \cite{Dep-69} implies that a perturbed
Hamiltonian system admits a global normalization of arbitrary
order relative to the unperturbed Hamiltonian vector field with
periodic flow.

In this paper, we are interested in a generalized version of problem
(\ref{Int1}), (\ref{Int2}), when $F,G$ and $\bar{G}$ are tensor
fields on $M$ of arbitrary type. In the general case, when
$T\neq\operatorname{const}$, the property like (\ref{Int3}) is no
longer true and as a consequence there are obstructions to the
solvability of (\ref{Int1}), (\ref{Int2}). Using some algebraic
properties of the $\mathbb{S}^{1}$averaging and calculus on exterior
algebras, we generalize the free-coordinate formulas (\ref{Int4}),
(\ref{Int5}) to the case of multivector fields and differential
forms on $M$. These results are applied to the normalization problem
for a special class of perturbed Hamiltonian dynamics on slow-fast
phases spaces \cite{DaVo-08,Da-Vo-09}, which leads to the study of
homological equations for non-Hamiltonian vector fields. Finally, we
show how the homological equations for 1-forms appear in the context
of Hamiltonization problem \cite{DaVo-08,Vor-05} and the
construction of symplectic structures which are invariant with
respect to skew-product $\mathbb{S}^1$-actions \cite{Vor-11}.

\section{Algebraic properties of the $\mathbb{S}^{1}$-averaging}

Let $M$ be a smooth manifold. Denote by $\mathcal{T}_{s}^{k}(M)$ the space of
all tensor fields on $M$ of type $(k,m)$. In particular, $\mathcal{T}_{0}%
^{0}(M)=C^{\infty}(M)$ and $\mathcal{T}_{0}^{1}(M)=$
$\mathfrak{X}(M)$ are the spaces of smooth functions and vector
fields on $M$, respectively. For every vector field
$X\in\mathfrak{X}(M)$, we
denote by $\mathcal{L}_{X}:\mathcal{T}_{s}^{k}(M)\rightarrow\mathcal{T}%
_{s}^{k}(M)$ the Lie derivative along $X$ , that is, the unique
differential operator on the tensor algebra of the manifold $M$
which coincides with the standard Lie derivative $\mathcal{L}_{X}$
on $C^{\infty}(M)$ and $\mathfrak{X}(M)$ (see, for example
\cite{AbMa-88}). The flow $\operatorname{Fl}_{X}^{t}$ of $X$ and the
Lie derivative $\mathcal{L}_{X}$ are related  by the formula
\[
\frac{d(\operatorname{Fl}_{X}^{t})^{\ast}\Xi}{dt}=(\operatorname{Fl}_{X}%
^{t})^{\ast}(\mathcal{L}_{Y}\Xi),
\]
for any $\Xi\in\mathcal{T}_{s}^{k}(M)$.

Now, suppose that we are given an action of the circle $\mathbb{S}%
^{1}=\mathbb{R}/2\pi\mathbb{Z}$ on $M$ with infinitesimal generator
$\Upsilon$. Therefore, $\Upsilon$ is a complete vector field on $M$
whose flow $\operatorname{Fl}_{\Upsilon}^{t}$ is $2\pi$-periodic. We
admit that the $\mathbb{S}^{1}$-action is not necessarily free.

For every tensor field $\Xi\in\mathcal{T}_{s}^{k}(M)$, its average
with respect to the $\mathbb{S}^{1}$-action is a tensor field
$\langle\Xi\rangle$ $\in\mathcal{T}_{s}^{k}(M)$ of the same type
which is defined as \cite{MaMoRa-90}
\begin{equation}
\langle\Xi\rangle:=\frac{1}{2\pi}\int_{0}^{2\pi}(\operatorname{Fl}_{\Upsilon}^{t})^{\ast
}\Xi dt. \label{Av}%
\end{equation}

A tensor field $\Xi\in\mathcal{T}_{s}^{k}(M)$ is said to be
invariant with respect to the $\mathbb{S}^{1}$-action if
$(\operatorname{Fl}_{\Upsilon}^{t})^{\ast}\Xi=\Xi$ $(\forall
t\in\mathbb{R})$ or, equivalently, $\mathcal{L}_{\Upsilon}\Xi=0$. In
terms of the $\mathbb{S}^{1}$-average of $\Xi$, the
$\mathbb{S}^{1}$-invariance condition reads $\Xi=
\langle\Xi\rangle.$ Denote by
$\mathcal{A}:\mathcal{T}_{s}^{k}(M)\rightarrow\mathcal{T}_{s}^{k}(M)$
the averaging operator, $\mathcal{A}(\Xi)=\langle\Xi\rangle$ which
is a $\mathbb{R}$-linear operator with property
$\mathcal{A}^{2}=\mathcal{A}$. It is clear that the image of
$\mathcal{A}$ consists of all $\mathbb{S}^{1}$-invariant tensor
fields. A tensor field belongs to $\operatorname{Ker}\mathcal{A}$ if
its $\mathbb{S}^{1}$-average is zero. Therefore, we have the
$\mathbb{S}^{1}$-invariant splitting
\begin{equation}
\mathcal{T}_{s}^{k}(M)=\operatorname{Im}\mathcal{A\oplus}\operatorname{Ker}
\mathcal{A}.\label{Split}%
\end{equation}
Introduce also the $\mathbb{R}$-linear operator $\mathcal{S}:\mathcal{T}_{s}^{k}(M)\rightarrow\mathcal{T}_{s}^{k}(M)$ given by
\begin{equation}
\mathcal{S}(\Xi):=\frac{1}{2\pi}\int_{0}^{2\pi}(t-\pi)(\operatorname{Fl}
_{\Upsilon}^{t})^{\ast}\Xi dt.\label{S}
\end{equation}
It follows directly from definitions that the operators $\mathcal{L}_{\Upsilon},\mathcal{A}$ and $\mathcal{S}$ pairwise commute and satisfy the relations
\begin{equation}
\mathcal{A}\circ\mathcal{L}_{\Upsilon}=\mathcal{L}_{\Upsilon}\circ
\mathcal{A}=0,\label{P3}
\end{equation}
\begin{equation}
\mathcal{A}\circ\mathcal{S}=\mathcal{S\circ A}=0.\label{P2}
\end{equation}
Moreover, we have the following important property.

\begin{proposition}
The following identity holds
\begin{equation}
\mathcal{L}_{\Upsilon}\circ\mathcal{S}=\operatorname*{id}-\mathcal{A},
\label{P1}
\end{equation}
\end{proposition}
\begin{proof}
For every tensor field $\Xi\in\mathcal{T}_{s}^{k}(M)$, by definition
(\ref{S}), we have
\begin{align*}
(\operatorname{Fl}_{\Upsilon}^{\tau})^{\ast}\mathcal{S}(\Xi) &  =\frac{1}
{2\pi}\int_{0}^{2\pi}(t-\pi)(\operatorname{Fl}_{\Upsilon}^{t+\tau})^{\ast}\Xi
dt\\
&  =\frac{1}{2\pi}\int_{\tau}^{2\pi+\tau}(t-\tau-\pi)(\operatorname{Fl}
_{\Upsilon}^{t})^{\ast}\Xi dt
\end{align*}
Differentiating the both sides of this equality in $\tau$ and using the $2\pi
$-periodicity of the flow $\operatorname{Fl}_{\Upsilon}^{\tau}$, we get
\begin{align*}
\frac{d}{d\tau}(\operatorname{Fl}_{\Upsilon}^{\tau})^{\ast}\mathcal{S}(\Xi) &
=\frac{1}{2\pi}[(t-\tau-\pi)(\operatorname{Fl}_{\Upsilon}^{t})^{\ast}\Xi
]\mid_{\tau}^{2\pi+\tau}-\frac{1}{2\pi}\int_{\tau}^{2\pi+\tau}
(\operatorname{Fl}_{\Upsilon}^{t})^{\ast}\Xi dt\\
&
=(\operatorname{Fl}_{\Upsilon}^{\tau})^{\ast}(\Xi-\langle\Xi\rangle)
\end{align*}
Comparing this equality with the identity
\begin{equation}
\frac{d}{d\tau}(\operatorname{Fl}_{\Upsilon}^{\tau})^{\ast}\mathcal{S}
(\Xi)=(\operatorname{Fl}_{\Upsilon}^{\tau})^{\ast}(\mathcal{L}_{\Upsilon
}\mathcal{S}(\Xi))\label{Fl1}
\end{equation}
gives
$\mathcal{L}_{\Upsilon}(\mathcal{S}(\Xi))=\Xi-\langle\Xi\rangle$.
\end{proof}

\begin{corollary}
For every tensor field $\Xi$ $\in\mathcal{T}_{s}^{k}(M)$, the following
assertions are equivalent

\begin{itemize}
\item $\mathcal{S}(\Xi)=0$;

\item $\mathcal{S}(\Xi)$ is $\mathbb{S}^{1}$-invariant;

\item $\Xi$ is $\mathbb{S}^{1}$-invariant.
\end{itemize}
\end{corollary}
\begin{proof}
The equivalence of the first two conditions follows from property
(\ref{P2}) which says that $\langle\mathcal{S}(\Xi)\rangle=0$.
Property (\ref{P1}) implies the equivalence of the last two
assertions.
\end{proof}

\begin{proposition}
The following relations hold
\begin{equation}
\operatorname{Ker}\mathcal{S}=\operatorname{Ker}\mathcal{L}_{\Upsilon
}=\operatorname{Im}\mathcal{A}, \label{P5}
\end{equation}
\begin{equation}
\operatorname{Im}\mathcal{S}=\operatorname{Im}\mathcal{L}_{\Upsilon
}=\operatorname{Ker}\mathcal{A}. \label{P6}
\end{equation}
\end{proposition}
\begin{proof}
Taking into account that the kernel of the Lie derivative $\mathcal{L}
_{\Upsilon}:\mathcal{T}_{s}^{k}(M)\rightarrow\mathcal{T}_{s}^{k}(M)$ consists
of all $\mathbb{S}^{1}$-invariant tensor fields and by the Corollary 2.2 ,
we derive (\ref{P5}). By (\ref{P3}), we have $\operatorname{Im}\mathcal{L}%
_{\Upsilon}\subseteq\operatorname{Ker}\mathcal{A}$. On the other hand, it
follows from (\ref{P1}) that
\begin{equation}
\mathcal{L}_{\Upsilon}\mathcal{S}(\Xi)=\Xi\text{ }\forall\Xi
\in\operatorname{Ker}\mathcal{A}\label{P7}%
\end{equation}
and hence $\operatorname{Ker}\mathcal{A\subseteq}\operatorname{Im}%
\mathcal{L}_{\Upsilon}$. Therefore, $\operatorname{Im}\mathcal{L}_{\Upsilon
}=\operatorname{Ker}\mathcal{A}$. By (\ref{P3})-(\ref{P1}) we have the
identities $\mathcal{S}=\mathcal{L}_{\Upsilon}\circ\mathcal{S}^{2}$ and
$\mathcal{L}_{\Upsilon}=\mathcal{S}\circ\mathcal{L}_{\Upsilon}^{2}$  which say
that $\operatorname{Im}\mathcal{S}=\operatorname{Im}\mathcal{L}_{\Upsilon}$.
\end{proof}

As a consequence of (\ref{Split}) and (\ref{P5}),(\ref{P6}), we get
also the decomposition
\begin{equation}
\mathcal{T}_{s}^{k}(M)=\operatorname{Ker}\mathcal{L}_{\Upsilon}\oplus
\operatorname{Im}\mathcal{L}_{\Upsilon}\label{P10}%
\end{equation}
which together with (\ref{P7}) implies that the restriction of
$\mathcal{L}_{\Upsilon}$ to $\operatorname{Im}\mathcal{L}_{\Upsilon}$ is an
isomorphism whose inverse is just $\mathcal{S}$.

\section{Homological equations associated to periodic flows}

Let $X$ be a complete vector field on a manifold $M$ with periodic flow. This
means that there exists a a smooth positive function $T:M\rightarrow
\mathbb{R}$, called a period function,  such that $\operatorname{Fl}%
_{X}^{t+T(m)}(m)=\operatorname{Fl}_{X}^{t}(m)$  for all $m\in M$ and
$t\in\mathbb{R}$. Introduce also the frequency function $\omega
:M\rightarrow\mathbb{R}$ given by $\omega=\frac{2\pi}{T}$. It is
clear that
$\omega$ is a first integral of $X$. The periodic flow $\operatorname{Fl}%
_{X}^{t}$ induces the $\mathbb{S}^{1}$-action on $M$ whose infinitesimal
generator is the vector field $\Upsilon=\frac{1}{\omega}X$ with $2\pi
$-periodic flow.

\textbf{Homological equations for $k$-vector fields.} Let $\chi^{k}(M)=\operatorname*{Sec}(\wedge^{k}TM)$
be the space of all $k$-multivector fields on $M$. In particular,
$\chi^{0}(M)=C^{\infty}(M)$ and $\chi^{1}(M)=\mathfrak{X}(M)$ . It
is clear that the operators $\mathcal{L}_{\Upsilon},\mathcal{S}$ and
$\mathcal{A}$ leave invariant the subspaces
$\chi^{k}(M)\subset\mathcal{T}_{0}^{k}(M)$. For, every $k$-vector
field $A\in\chi^{k}(M)$ and a 1-form $\alpha$ on $M$, denote by $\mathbf{i}%
_{\alpha}A\in\mathcal{T}_{0}^{k-1}(M)$ a $(k-1)$-vector field
defined by
\[
(\mathbf{i}_{\alpha}A)(\alpha_{1},...,\alpha_{k-1})=A(\alpha,\alpha
_{1},...,\alpha_{k-1})
\]
for all 1-forms $\alpha_{1},...,\alpha_{k-1}$ on $M$. By definition,
$\mathbf{i}_{\alpha}A=0$, for every $0$-vector field $A$.

Consider the $\mathbb{S}^{1}$-action on $M$ associated to the periodic flow
$\operatorname{Fl}_{X}^{t}$. Denote by $\chi_{\operatorname{inv}}%
^{k}(M)=\operatorname{Ker}\mathcal{L}_{\Upsilon}$ the subspace of
all $\mathbb{S}^{1}$-invariant $k$-vector fields on $M$. Then,
according to (\ref{P10}), we have the splitting
\begin{equation}
\chi^{k}(M)=\chi_{\operatorname{inv}}^{k}(M)\oplus\chi_{0}^{k}(M),\label{DEC1}%
\end{equation}
where $\chi_{0}^{k}(M)=\operatorname{Im}\mathcal{L}_{\Upsilon}$
denotes the subspace of all $k$-vector fields on $M$ with zero
average.

\begin{theorem}\label{kvect}
Let $X$ be a vector field on $M$ with periodic flow and frequency function
$\omega$. Then, for a given $B\in\chi^{k}(M)$, all $k$-vector fields $A$ and
$\bar{B}$ on $M$ satisfying the homological equation
\begin{equation}
\mathcal{L}_{X}A=B-\bar{B}\label{H1}
\end{equation}
and the condition
\begin{equation}
\bar{B}\text{ is }\mathbb{S}^{1}\text{-invariant}\label{H2}
\end{equation}
are of the form
\begin{equation}
\bar{B}=\langle
B\rangle+\frac{1}{\omega}X\wedge\mathbf{i}_{d\omega}C,\label{S1}
\end{equation}
\begin{equation}
A=\frac{1}{\omega}\mathcal{S}(B)+\frac{1}{\omega^{3}}X\wedge\mathcal{S}%
^{2}(\mathbf{i}_{d\omega}B)+C,\label{S2}
\end{equation}
where $C\in\chi_{\operatorname{inv}}^{k}(M)$ is an arbitrary
$\mathbb{S}^{1}$-invariant $k$-vector field. Here, the average
$\langle\rangle$ is taken with respect to the
$\mathbb{S}^{1}$-action on $M$ associated to the flow of $X$.
\end{theorem}
\begin{proof}
Using the identity, \cite{AbMa-88}
\begin{equation}
\mathcal{L}_{\omega\Upsilon}A=\omega\mathcal{L}_{\Upsilon}A-\Upsilon
\wedge\mathbf{i}_{d\omega}A,\label{S3}
\end{equation}
we rewrite equation (\ref{H1}) in the form
\begin{equation}
\mathcal{L}_{\Upsilon}A-\frac{1}{\omega}\Upsilon\wedge\mathbf{i}_{d\omega
}A=\frac{1}{\omega}(B-\bar{B}).\label{PSV1}
\end{equation}
Applying the averaging operator to the both sides of this equation and taking
into account condition (\ref{H2}), we get
\begin{equation}
\bar{B}=\langle\bar{B}\rangle=\langle
B\rangle+\Upsilon\wedge\mathbf{i}_{d\omega}\langle
A\rangle.\label{PSV2}
\end{equation}
According to decomposition (\ref{DEC1}), we have
\begin{equation}
A=\langle A\rangle+A_{0},\ \ \ \ \langle A_{0}\rangle=0,\label{PSV3}
\end{equation}
\[
B=\langle B\rangle+B_{0},\ \ \ \ \langle B_{0}\rangle=0.
\]
Putting these representations together with (\ref{PSV2}) into (\ref{PSV1}), we
see that the original problem (\ref{H1}), (\ref{H2}) is reduced to the
following equation for $A_{0}\in\chi_{0}^{k}(M)$:
\[
\mathcal{L}_{\Upsilon}A_{0}-\frac{1}{\omega}\Upsilon\wedge\mathbf{i}_{d\omega
}A_{0}=\frac{1}{\omega}B_{0}.
\]
Looking for $A_{0}$ in the form $A_{0}=\frac{1}{\omega}\mathcal{S}%
(B_{0})+\tilde{A}_{0}$ and using property (\ref{P1}), we conclude that
$\tilde{A}_{0}\in\chi_{0}^{k}(M)$ must satisfy the equation
\begin{equation}
\mathcal{L}_{\Upsilon}\tilde{A}_{0}-\frac{1}{\omega}\Upsilon\wedge
\mathbf{i}_{d\omega}\tilde{A}_{0}=\frac{1}{\omega^{2}}\Upsilon\wedge
\mathbf{i}_{d\omega}\mathcal{S}(B_{0})\label{PSV5}
\end{equation}
Next, taking into account that $\mathbf{i}_{d\omega}\Upsilon=0$ and
putting $\tilde{A}_{0}=\Upsilon\wedge\mathcal{S}(D)$, we reduce
(\ref{PSV5}) to the following equation for $D\in\chi_{0}^{k-1}(M)$:
\[
\Upsilon\wedge\mathcal{L}_{\Upsilon}\mathcal{S}(D)=\frac{1}{\omega^{2}
}\Upsilon\wedge\mathbf{i}_{d\omega}\mathcal{S}(B_{0}).
\]
By property (\ref{P1}) the tensor field $D=\frac{1}{\omega^{2}}\mathcal{S}
^{2}(\mathbf{i}_{d\omega}B_{0})$ satisfies the relation $\mathcal{L}
_{\Upsilon}\mathcal{S}(D)=\frac{1}{\omega^{2}}\mathbf{i}_{d\omega}%
\mathcal{S}(B_{0})$ . Therefore, the solutions to problem (\ref{H1}),
(\ref{H2}) are given by (\ref{PSV2}) and (\ref{PSV3}), where
\begin{equation}
A_{0}=\frac{1}{\omega}\mathcal{S}(B_{0})+\frac{1}{\omega^{3}}X\wedge
\mathcal{S}^{2}(\mathbf{i}_{d\omega}B_{0})\label{PSV6}
\end{equation}
and $\langle A\rangle$ is an arbitrary $\mathbb{S}^{1}$-invariant
$k$-vector field on $M$. Finally, property (\ref{P2}) says that
$\mathcal{S}(B_{0})=\mathcal{S}(B)$ and hence formulas (\ref{S1})
and (\ref{S2}) follow from (\ref{PSV2}) and (\ref{PSV6}) with
$C=\langle A\rangle$.
\end{proof}

As a straightforward consequence of Theorem \ref{kvect}, we get the following result.

\begin{corollary}\label{corol}
The kernel of the Lie derivative $\mathcal{L}_{X}:$ $\chi^{k}(M)\rightarrow$
$\chi^{k}(M)$ is
\begin{equation}
\operatorname{Ker}(\mathcal{L}_{X})=\chi_{\operatorname{inv}}^{k}(M)\cap
\{C\in\chi^{k}(M)\mid X\wedge\mathbf{i}_{d\omega}C=0\}.\label{KerLX}%
\end{equation}
For a given $B\in\chi^{k}(M)$ $(k\geq1)$, the homological equation
\begin{equation}
\mathcal{L}_{X}A=B\label{Hom1}
\end{equation}
is solvable relative to a $k$-vector field $A$ on $M$ if and only if
\begin{equation}
\langle B\rangle=X\wedge\mathbf{i}_{d\omega}P\label{Hom2}
\end{equation}
for a certain $\mathbb{S}^{1}$-invariant $k$-vector field $P\in
\chi_{\operatorname{inv}}^{k}(M)$. Under this condition, every
solution of (\ref{Hom1}) is given by (\ref{S2}), where $C = P + C'$,
$C' \in \mathrm{ker}(\mathcal{L}_X)$.
\end{corollary}

It follows from (\ref{Hom2}) that the necessary conditions for the
solvability of (\ref{Hom1}) are the following
\begin{equation}
X(m)=0 \ \ \rightarrow \ \ \langle B\rangle(m)=0,\label{Nec1}%
\end{equation}
\begin{equation}
X\wedge\langle B\rangle=0,\label{Nec2}
\end{equation}
\begin{equation}
\mathbf{i}_{d\omega}\langle B\rangle=0.\label{Nec3}
\end{equation}
Therefore, if one of these conditions does not hold, then equation
(\ref{Hom1}) is unsolvable.

\begin{corollary}\label{corol2}
There exist $k$-vector fields $A$ and $\bar{B}$ on $M$ satisfying the
equations
\begin{equation}
\mathcal{L}_{X}A=B-\bar{B}, \label{Co1}%
\end{equation}%
\begin{equation}
\mathcal{L}_{X}\bar{B}=0 \label{Co2}%
\end{equation}
if and only if
\begin{equation}
X\wedge\mathbf{i}_{d\omega}\langle B\rangle=0. \label{Co3}%
\end{equation}
Under this condition, all solutions $(A,\bar{B})$ to
(\ref{Co1}),(\ref{Co2}) are given by formulas  (\ref{S1}),
(\ref{S2}). Moreover, the $k$-vector field $A$ in (\ref{S2}) can be
represented in the form $\displaystyle
A=\frac{1}{\omega}\mathcal{S}(B)+( \mathbb{S}^{1}$-invariant
k-vector field $)$  if and only if
$\Upsilon\wedge\mathbf{i}_{d\omega}B=0$.
\end{corollary}

Let us consider some particular cases. In the case $k=0$, formulas
(\ref{S1}) and (\ref{S2}) coincide with (\ref{Int4}), (\ref{Int5}).
The well-known solvability condition for the equation
$\mathcal{L}_{X}F=G$ reads $\langle G\rangle=0$.

In the case $k=1$, by Theorem \ref{kvect} , Corollary \ref{corol}
and Corollary \ref{corol2} (where $A=Z$ and $B=W$ are vector fields
on $M$) we derive the following facts. For a given
$W\in\mathfrak{X}(M)$, all vector fields $Z\in\mathfrak{X}(M)$ and
$\bar{W}\in\mathfrak{X}_{\operatorname{inv}}(M)$ on $M$ satisfying
the equation
\begin{equation}
[ X,Z]=W-\bar{W},\label{IL1}
\end{equation}
are given by the formulas
\begin{equation}
\bar{W}=\langle
W\rangle+\frac{1}{\omega}\mathcal{L}_{Y}(\omega)X,\label{IL3}
\end{equation}
\begin{equation}
Z=\frac{1}{\omega}\mathcal{S}(W)+\frac{1}{\omega^{3}}\mathcal{S}
^{2}(\mathcal{L}_{W}\omega)X+Y,\label{IL4}
\end{equation}
where $Y\in\mathfrak{X}_{\operatorname{inv}}(M)$. The second term in
(\ref{IL4}) can be omitted only if $\mathcal{L}_{W}\omega=0.$
Moreover, the homological equation
\begin{equation}
[ X,Z]=W\label{IL6}
\end{equation}
for $Z$ is solvable if and only if $\langle
W\rangle=\mathcal{L}_{\tilde{Y}}(\omega)X$ for a certain
$\mathbb{S}^{1}$-invariant vector field $\tilde{Y}$. In particular ,
the necessary condition (\ref{Nec3})  for the solvability of
(\ref{IL6}) takes the form
\begin{equation}
\mathcal{L}_{\langle W\rangle}\omega=0.\label{IL11}
\end{equation}
Let $\operatorname{Reg}(X)=\{m\in M\mid X(m)\neq0\}$ be the set of
points regular of $X$. If $\operatorname{Reg}(X)$ is everywhere
dense in $M$, then the kernel of the Lie derivative
$\mathcal{L}_{X}:$ $\mathfrak{X}(M)\rightarrow$ $\mathfrak{X}(M)$ is
\begin{equation}
\operatorname{Ker}(\mathcal{L}_{X})=\mathfrak{X}_{\operatorname{inv}}
(M)\cap\{Y\in\mathfrak{X}(M)\mid\mathcal{L}_{Y}\omega=0\}.\label{IL5}%
\end{equation}
Moreover, it follows from (\ref{Co3}) that there exist vector fields $Z$ and
$\bar{W}$ on $M$ satisfying homological equation (\ref{IL1}) and the condition
$[X,\bar{W}]=0$ if and only if condition (\ref{IL11}) holds.

\textbf{Homological equations for $k$-Forms.} Consider the space
$\Omega^{k}(M)=\operatorname*{Sec}(\wedge^{k}T^{\ast}M)$ of $k$-forms on $M$.
Then, subspace $\Omega^{k}(M)\subset\mathcal{T}_{k}^{0}(M)$ is an
invariant with respect to the action of the operators $\mathcal{L}_{\Upsilon
},\mathcal{S}$ and $\mathcal{A}$. By $\mathbf{i}_{Y}\alpha$ $\in\Omega
^{k-1}(M)$ we denote the interior product of a vector field $Y$ and a $k$-form
$\alpha$ on $M$ which is defined by the usual formula: $(\mathbf{i}_{Y}%
\alpha)(Y_{1},...,Y_{k-1})=\alpha(Y,Y_{1},...,Y_{k-1})$. Let $\Omega
_{\operatorname{inv}}^{k}(M)=\operatorname{Ker}\mathcal{L}_{\Upsilon}$
and $\Omega_{0}^{k}(M)=\operatorname{Im}\mathcal{L}_{\Upsilon}$.
Then, we have the $\mathbb{S}^{1}$-invariant splitting
\begin{equation}
\Omega^{k}(M)=\Omega_{\operatorname{inv}}^{k}(M)\oplus\Omega_{0}.
^{k}(M),\label{DEC2}%
\end{equation}
There is the following covariant analog of Theorem \ref{kvect}.

\begin{theorem}\label{homolk}
For a given $\eta\in\Omega^{k}(M)$, all $k$-forms $\theta$ and $\bar{\eta}$ on
$M$ satisfying the homological equation
\begin{equation}
\mathcal{L}_{X}\theta=\eta-\bar{\eta} \label{Fom1}%
\end{equation}
and the condition
\begin{equation}
\bar{\eta}\text{ is }\mathbb{S}^{1}\text{-invariant} \label{Fom2},
\end{equation}
are represented as
\begin{equation}
\bar{\eta}=\langle \eta\rangle-\frac{1}{\omega}d\omega\wedge\mathbf{i}_{X}\mu, \label{Fom3}%
\end{equation}%
\begin{equation}
\theta=\frac{1}{\omega}\mathcal{S}(\eta)-\frac{1}{\omega^{3}}d\omega
\wedge\mathcal{S}^{2}(\mathbf{i}_{X}\eta)+\mu, \label{Fom4}%
\end{equation}
where $\mu$ $\in\Omega_{\operatorname{inv}}^{k}(M)$ is an arbitrary
$\mathbb{S}^{1}$-invariant $k$-form.
\end{theorem}

The proof of this theorem goes in the same line as the proof Theorem \ref{kvect},
where instead of identity (\ref{S3}) we have to use its covariant analog:
$\mathcal{L}_{\omega\Upsilon}\theta=\omega\mathcal{L}_{\Upsilon}%
\theta+d\omega\wedge\mathbf{i}_{\Upsilon}\theta$.

Let $\eta=\eta_{0}+\langle \eta\rangle$. Then, the general solution
$\theta$ in (\ref{Fom4}) has the representation $\theta=\theta_{0}$
$+\mu$, where
$\theta_{0}\in\Omega_{0}^{k}(M)$ is uniquely determined by $\eta_{0}$,%
\begin{equation}
\theta_{0}==\frac{1}{\omega}\mathcal{S}(\eta_{0})-\frac{1}{\omega^{3}}%
d\omega\wedge\mathcal{S}^{2}(\mathbf{i}_{X}\eta_{0})\label{Fom5}%
\end{equation}

From Theorem \ref{homolk}, we deduce the following consequences

\begin{corollary}\label{corolker}
The kernel of the Lie derivative $\mathcal{L}_{X}: \Omega^{k}(M)\rightarrow\Omega^{k}(M)$ is
\begin{equation}
\operatorname{Ker}(\mathcal{L}_{X})=\Omega_{\operatorname{inv}}^{k}%
(M)\cap\{\mu\in\Omega^{k}(M)\mid d\omega\wedge\mathbf{i}_{X}\mu=0\}.
\label{Fom10}
\end{equation}
For a given $\eta\in\Omega^{k}(M)$ $(k\geq1)$, the homological equation
\begin{equation}
\mathcal{L}_{X}\theta=\eta\label{Fom11}
\end{equation}
is solvable relative to a $k$-form $\theta$ on $M$ if and only if
\begin{equation}
\langle \eta\rangle=d\omega\wedge\mathbf{i}_{X}\alpha\label{Fom12}%
\end{equation}
for a certain $\mathbb{S}^{1}$-invariant $k$-form $\alpha\in\Omega
_{\operatorname{inv}}^{k-1}(M)$.
\end{corollary}

It follows from (\ref{Fom12}) that the necessary conditions for the
solvability of equation (\ref{Fom11}) are
\begin{equation}
X(m)=0\Longrightarrow\text{ }\langle \eta\rangle(m)=0, \label{Fom13}
\end{equation}
\begin{equation}
d\omega\wedge\langle \eta\rangle=0, \label{Fom14}
\end{equation}
\begin{equation}
\mathbf{i}_{X}\langle \eta\rangle=0. \label{Fom15}
\end{equation}

\begin{corollary}\label{closedcorol}
There exist $k$-forms $\theta$ and $\eta$ on $M$ satisfying the equations
\begin{equation}
\mathcal{L}_{X}\theta=\eta-\bar{\eta}, \label{Fom16}%
\end{equation}%
\begin{equation}
\mathcal{L}_{X}\bar{\eta}=0, \label{Fom17}%
\end{equation}
if and only if the following condition holds%
\begin{equation}
d\omega\wedge\mathbf{i}_{X}\langle \eta\rangle=0 \label{Fom18}%
\end{equation}
\end{corollary}

In the case $k=1$, for a given 1-form $\eta\in$ $\Omega^{1}(M)$,
formulas (\ref{Fom3}), (\ref{Fom4}) for solutions
$(\theta,\bar{\eta})$ of problem (\ref{Fom1}), (\ref{Fom2}) can be
written as follows
\begin{equation}
\bar{\eta}=\langle \eta\rangle-\frac{1}{\omega}(\mathbf{i}_{X}\mu)d\omega,\label{Fom19}%
\end{equation}%
\begin{equation}
\theta=\frac{1}{\omega}\mathcal{S}(\eta)-\frac{1}{\omega^{3}}\mathcal{S}%
^{2}(\mathbf{i}_{X}\eta)d\omega+\mu,\label{Fom20}%
\end{equation}
with $\mu$ $\in\Omega_{\operatorname{inv}}^{1}(M)$. The solvability condition
for homological equation (\ref{Fom11}) reads%
\begin{equation}
\langle \eta\rangle=(\mathbf{i}_{X}\alpha)d\omega\label{CSP}%
\end{equation}
for a certain $\mathbb{S}^{1}$-invariant 1-form $\alpha\in\Omega
_{\operatorname{inv}}^{1}(M)$. Let $\operatorname{Reg}(\omega)=\{m\in M\mid
d_{m}\omega\neq0\}$ be the set of regular points of the frequency function
$\omega$. If $\operatorname{Reg}(\omega)$ is everywhere dense in
$M$, then
\begin{equation}
\operatorname{Ker}(\mathcal{L}_{X})=\Omega_{\operatorname{inv}}^{1}%
(M)\cap\{\mu\in\Omega^{k}(M)\mid\mathbf{i}_{X}\mu=0\}.\label{Fom21}%
\end{equation}
Moreover, condition (\ref{Fom18}) for the solvability of problem
(\ref{Fom16}), (\ref{Fom17}) is equivalent to the following
$\mathbf{i}_{X}\langle \eta\rangle=0$.

To end this section, let us consider the case of closed forms. Let
$d:\Omega^{k}(M)\rightarrow\Omega^{k+1}(M)$ be the exterior
derivative. By standard properties of the exterior derivative, we
conclude that $d$ commutes with operators
$\mathcal{L}_{\Upsilon},\mathcal{S}$ and $\mathcal{A}$, in
particular, $\langle d\eta\rangle=d\langle \eta\rangle$ for any
$\eta\in\Omega^{k}(M)$. Moreover, splitting (\ref{DEC2}) is
invariant with respect to $d$, in the sense that
\[
d\eta=\langle d\eta\rangle+(d\eta)_{0}=d\langle
\eta\rangle+d\eta_{0},
\]
where $\eta=\langle \eta\rangle+\eta_{0}$. It follows that if $\eta$
is closed then, the components $\langle \eta\rangle$ and $\eta_{0}$
are also closed $k$-forms. In this case, according to (\ref{Fom5}),
a solution to the equation $\mathcal{L}_{\Upsilon
}\theta_{0}=\eta_{0}$ is given by
$\theta_{0}=\mathcal{S}(\eta_{0})$. Then,
$d\theta_{0}=\mathcal{S}(d\eta_{0})=0$ and hence $\eta_{0}=\mathcal{L}%
_{\Upsilon}\theta_{0}=d\circ\mathbf{i}_{\Upsilon}\theta_{0}$. This proves the
following assertion.

\begin{proposition}
For every closed $k$-form $\eta$ on $M$, we have the decomposition
\begin{equation}
\eta=\langle \eta\rangle+d(\mathbf{i}_{\Upsilon}\theta_{0}),
\label{Clos1}
\end{equation}
where $\theta_{0}=\frac{1}{2\pi}\int_{0}^{2\pi}(t-\pi)(\operatorname{Fl}
_{\Upsilon}^{t})^{\ast}\eta dt$.
\end{proposition}

\section{Normal forms}

Suppose we start again with a vector field $X$ on $M$ whose flow is periodic
with frequency function $\omega:M\rightarrow\mathbb{R}$.

\begin{proposition}\label{propnorm}
Let $P_{\varepsilon}=X+\varepsilon W$ be a perturbed vector field on $M$,
where $\varepsilon$ is a small parameter. Let
\begin{equation}
\Phi_{\varepsilon}=\operatorname{Fl}_{Z}^{t}\mid_{t=\varepsilon} \label{NIT}
\end{equation}
be the time-$\varepsilon$ flow of the vector field
\begin{equation}
Z=\frac{1}{\omega}\mathcal{S}(W)+\frac{1}{\omega^{3}}\mathcal{S}
^{2}(\mathcal{L}_{W}\omega)X+Y, \label{ED1}
\end{equation}
where $Y$ is an $\mathbb{S}^{1}$-invariant vector field on $M$.
Then, for a given open domain $N\subset M$ with compact closure,
there exists a constant $\delta >0$ such that formula
(\ref{NIT}) defines a near identity transformation
$\Phi_{\varepsilon}:N\rightarrow M$ for $\varepsilon\in
(-\delta,\delta)$ which brings $P_{\varepsilon}$ into the form
\begin{equation}
(\Phi_{\varepsilon})^{\ast}{}P_{\varepsilon}=X+\varepsilon\bar{W}%
+O(\varepsilon^{2}) \label{ED2}%
\end{equation}%
\begin{equation}
\bar{W}:=\langle W\rangle+\mathcal{L}_{Y}(\ln\omega)X. \label{ED7}%
\end{equation}
If the frequency function $\omega$ is a first integral of the $\mathbb{S}^{1}
$-average of $W$,
\begin{equation}
\mathcal{L}_{\langle W\rangle}\omega=0\text{ on }M,  \label{ED3}
\end{equation}
then the mapping $\Phi_{\varepsilon}$ is a normalization transformation of
first order for $P_{\varepsilon}$ relative to $X,$
\begin{equation}
\lbrack X,\bar{W}]=0, \label{ED4}%
\end{equation}
for arbitrary choice of a vector field $Y\in\mathfrak{X}_{\operatorname*{inv}%
}(M)$ in (\ref{ED1}).
\end{proposition}

The proof of this proposition follows from the standard Lie
transform \cite{Dep-69} arguments and Corollary \ref{closedcorol} for the case
$k=1$. Remark that if $\operatorname{Reg}(X)$ is everywhere dense in
$M$, then (\ref{ED3}) becomes also a necessary condition for mapping
$\Phi_{\varepsilon}$ (\ref{NIT}) to be a normalization
transformation.

Now, let us see how, in the context of the normalization procedure,
one can use a freedom in the definition of $\Phi _{\varepsilon}$.
Consider the perturbed vector field $P_{\varepsilon }=X+\varepsilon
W$ and assume that the $\mathbb{S}^{1}$-action with infinitesimal
generator $\Upsilon$ $=\frac{1}{\omega}X$ is free on $M$. Then, the
orbit space $\mathcal{O}=M /\mathbb{S}^{1}$ is a smooth manifold and
the projection $\rho:M\rightarrow$ $\mathcal{O}$ is a $\mathbb{S}^{1}%
$-principle bundle. In this case, the frequency function is of the form
$\omega=\omega_{\mathcal{O}}\circ\rho$ for a certain $\omega_{\mathcal{O}}\in
C^{\infty}(\mathcal{O})$. Let $\operatorname{Ver}=\operatorname{Span}%
\{\Upsilon\}$ be the vertical subbundle and $\mathcal{D}\subset TM$
an arbitrary subbundle which is complimentary to
$\operatorname{Ver}$. Then, for every vector field
$u\in\mathfrak{X}(\mathcal{O})$ there exists a unique
$e\in\operatorname{Sec}(\mathcal{D})$ descending to $u$,
$d\rho\circ$ $e=u\circ\rho$. It follows that $[\Upsilon
,e]=b\Upsilon$, where $b\in C^{\infty}(M)$ with $\langle
b\rangle=0$. Defining $\operatorname{hor}(u):=e-\mathcal{S}(b)e$ ,
by property (\ref{P1}), we get that
$[\Upsilon,\operatorname{hor}(u)]=0$ . Therefore, we have the
splitting $TM=\operatorname{Hor}\oplus\operatorname{Ver}$ (a
principle connection on $M$),  where the horizontal subbundle
$\operatorname{Hor}=\operatorname*{Span}\{\operatorname{hor}(u)\mid
u\in\mathfrak{X}(\mathcal{O})\}$ is invariant with respect to the
$\mathbb{S}^{1}$-action ( for more details, see \cite{MaMoRa-90}).
According to this splitting, the vector field $\bar{W}$ in
(\ref{ED7}) has the decomposition $\bar{W}=\bar
{W}^{\operatorname{hor}}+\bar{W}^{\operatorname{ver}}$ into
horizontal and vertical parts. The following statement shows that
under an appropriate choice of
$Y\in\mathfrak{X}_{\operatorname*{inv}}(M)$, we can get $\bar
{W}^{\operatorname{ver}}=0$.

\begin{proposition}\label{nocrit}
If
\begin{equation}
d\omega\neq0\text{ on }M, \label{Nond}
\end{equation}
then one can choose an $\mathbb{S}^{1}$-invariant vector field $Y$
(\ref{Vert3}) in a such way that the near identity transformation
$\Phi_{\varepsilon}$ (\ref{NIT}) brings the perturbed vector field
$P_{\varepsilon}=X+\varepsilon W$ into the form $\tilde{P}_{\varepsilon
}=(\Phi_{\varepsilon})^{\ast}{}P_{\varepsilon}=\tilde{P}_{\varepsilon
}^{\operatorname{hor}}+\tilde{P}_{\varepsilon}^{\operatorname{ver}}$ with
\begin{equation}
\tilde{P}_{\varepsilon}^{\operatorname{ver}}=X+O(\varepsilon^{2}),
\label{Ave1}
\end{equation}
\begin{equation}
\tilde{P}_{\varepsilon}^{\operatorname{hor}}=\varepsilon\operatorname{hor}
(w)+O(\varepsilon^{2}), \label{Ave2}
\end{equation}
where $w\in\mathfrak{X}(\mathcal{O})$ is a unique vector field such
that $d\rho\circ\langle W\rangle=w\circ\rho$.
\end{proposition}
\begin{proof}
First, let us assume that $\mathcal{O}$ is parallelizable and pick a
basis of global vector fields $u_{1},...,u_{n}$ on $\mathcal{O}$.
Then, we have the basis of global $\mathbb{S}^{1}$-invariant vector
fields $\Upsilon,$
$\operatorname{hor}(u_{1}),...,\operatorname{hor}(u_{n})$ on $M$.
For the perturbation vector field $W$, we have the decomposition
$W=W^{\operatorname{hor}}+W^{\operatorname{ver}}$, where
$W^{\operatorname{hor}}=\sum_{i=1}^{n}c_{i}\operatorname{hor}(u_{i})$
and $W^{\operatorname{ver}}=c_{0}\Upsilon$ for some $c_{i}\in
C^{\infty}(M)$. Then, its $\mathbb{S}^{1}$-average is given by
\[
\langle W\rangle=\sum_{i=1}^{n}\langle
c_{i}\rangle\operatorname{hor}(u_{i})+\langle c_{0}\rangle\Upsilon
\]
It follows that the condition $\bar{W}^{\operatorname{ver}}=0$ is
equivalent to the algebraic equation $\mathbf{i}_{Y}d\omega=-\langle
c_{0}\rangle$ for $Y\in \mathfrak{X}_{\operatorname*{inv}}(M)$.
Under assumption (\ref{Nond}), a solution to this equation is
\begin{equation}
Y=-\frac{\langle c_{0}\rangle}{a^{2}}\sum_{i=1}^{n}a_{i}\operatorname{hor}(u_{i}%
),\label{Vert3}
\end{equation}
where $a_{i}=\mathbf{i}_{\operatorname{hor}(u_{i})}d\omega$ are $\mathbb{S}%
^{1}$-invariant functions on $M$ and $a^{2}=\sum_{i=1}^{n}a_{i}^{2}$. In the
general case, the statement follows from the partition of unity argument.
\end{proof}

Remark that in terms of the averaged vector field $w$ the
normalization condition (\ref{ED3}) reads
$\mathcal{L}_{w}\omega_{\mathcal{O}}=0$ on $\mathcal{O}$. In this
case, $[X,\operatorname{hor}(w)]=0$.

\textbf{The Hamiltonian case}. Let us show that, in the case when
the perturbed vector field is Hamiltonian, the normalization
condition (\ref{ED3}) is satisfied. Let $(M,\Omega)$ be a symplectic
manifold. Suppose that a perturbed vector field
$P_{\varepsilon}=X+\varepsilon W$ is Hamiltonian relative to the
symplectic structure $\Omega$, $\mathbf{i}_{X}\Omega=-dH_{0}$ and
$\mathbf{i}_{W}\Omega=-dH_{1}$ for some $H_{0},H_{1}\in
C^{\infty}(M)$. Assume that the flow of $X$ is periodic with
frequency function $\omega$ and
the corresponding $\mathbb{S}^{1}$-action is free. In particular $dH_{0}%
\neq0$ on $M$. Then, according to the period-energy relation for
Hamiltonian systems \cite{Gor-69}, \cite{Cush-84}, we have the
identity
\begin{equation}
dH_{0}\wedge d\omega=0\label{Vert4}%
\end{equation}
saying that $\omega$ functionally depends only on $H_{0}$. It
follows from (\ref{Vert4}) that the $\mathbb{S}^{1}$-action
preserves the symplectic form
and hence the averaged vector field $\langle W\rangle$ is also Hamiltonian, $\mathbf{i}%
_{\langle W\rangle}\Omega=-d\langle H_{1}\rangle$. Then,
\[
\mathbf{i}_{\langle W\rangle}dH_{0}=\Omega(\langle W\rangle,X)=-\mathbf{i}_{X}d\langle H_{1}\rangle=-\mathcal{L}%
_{X}\langle H_{1}\rangle=0.
\]
Finally, from here and (\ref{Vert4}) we get the equality
\[
0=(\mathbf{i}_{\langle W\rangle}dH_{0})d\omega-(\mathbf{i}_{\langle
W\rangle}d\omega)dH_{0}
=-(\mathcal{L}_{\langle W\rangle}\omega)dH_{0}%
\]
which implies (\ref{ED3}). Moreover, one can show that formula
(\ref{ED1}) for $Y=0$ gives a vector field $Z$ which is Hamiltonian
relative to $\Omega$ and the function
$\mathcal{S}(\frac{H_{1}}{\omega})$. Therefore, in the Hamiltonian
case, Proposition \ref{propnorm}  leads to the well-known result
\cite{Cush-84} on the global normalization of perturbed Hamiltonian
dynamics relative to periodic flows.

\section{The averaging on slow-fast phase spaces}

Let $M=M_{1}\times M_{2}$ be a product of two symplectic manifolds
$(M_{1},\sigma_{1})$ and $(M_{2},\sigma_{2})$. Let
$\pi_{1}:M\rightarrow M_{1}$ and $\pi_{2}:M\rightarrow M_{2}$ be the
canonical projections and $d_{1}$ and $d_{2}$ the partial exterior
derivatives on $M$ along $M_{1}$ and $M_{2}$, respectively. It is
clear that $d=d_{1}+d_{2}$ is the exterior derivative on $M$ and
$d_{1}^{2}=d_{2}^{2}=d_{1}\circ d_{2}+d_{2}\circ d_{1}=0$. Introduce
the following $\varepsilon$-dependent 2-form on $M$:
\begin{equation}
\sigma=\pi_{1}^{\ast}\sigma_{1}+\varepsilon\pi_{2}^{\ast}\sigma_{2}\label{PSF1}
\end{equation}
which is a symplectic structure for all $\varepsilon\neq 0$. For
$H\in C^{\infty}(M)$, denote by $V_{H}$ the Hamiltonian vector field
relative to $\sigma$. Then,
$V_H=V_H^{(1)}+\frac{1}{\varepsilon}V_H^{(2)}$, where $V_H^{(1)}$
and $V_H^{(2)}$ are vector fields on $M$ uniquely defined by the
relations
\begin{equation}
\mathbf{i}_{V_H^{(1)}}(\pi_{1}^{\ast}\sigma_{1})=-d_{1}H,\label{PSF2}%
\end{equation}
\begin{equation}
d\pi_{2}\circ V_H^{(1)}=0\label{PSF3}
\end{equation}
and
\begin{equation}
\mathbf{i}_{V_H^{(2)}}(\pi_{2}^{\ast}\sigma_{2})=-d_{2}H,\label{PSF4}
\end{equation}%
\begin{equation}
d\pi_{1}\circ V_H^{(2)}=0.\label{PSF5}
\end{equation}
It follows that, for all $m_{1}\in M_{1}$ and $m_{2}\in M_{2}$, the
vector fields $V_H^{(1)}$ and $V_H^{(2)}$ are tangent to the
symplectic slices $M_{1}\times\{m_{2}\}$ and $\{m_{1}\}\times
M_{2}$, respectively. For every $u\in\mathfrak{X}(M_{1})$, denote by
$\hat{u}=u\oplus0\in\mathfrak{X}(M)$ the lifting associated to the
canonical decomposition $TM=TM_{1}\oplus TM_{2}$.

On the slow-fast phase space $(M,\sigma)$, let us consider the
following perturbed Hamiltonian model
\cite{DaVo-08,Da-Vo-09,Vor-11,Vor-05}
\begin{equation}
H_{\varepsilon}=f\circ\pi_{1}+\varepsilon F, \label{HamF}%
\end{equation}
for some $f\in C^{\infty}(M_{1})$ and $F\in C^{\infty}(M)$. The
corresponding Hamiltonian vector field takes the form
\begin{equation}
V_{H_{\varepsilon}}=\mathbb{V}+\varepsilon\mathbb{W}, \label{Vec1}
\end{equation}
where
\begin{equation}
\mathbb{V}=\hat{v}_{f}+V_F^{(2)} \label{Vec2}%
\end{equation}
and
\begin{equation}
\mathbb{W}=V_F^{(1)} \label{Vec3}%
\end{equation}
are unperturbed and perturbation vector fields, respectively. Here
$v_{f}$ denotes the Hamiltonian vector field on
$(M_{1},\sigma_{1})$ of $f$.

Remark that, in general, the unperturbed vector field $\mathbb{V}$
is not Hamiltonian relative to the symplectic structure
(\ref{PSF1}). Indeed, it is easy to show that this happens only if
$F=\pi_{1}^{\ast}f_{1}+\pi_{2}^{\ast
}f_{2}$, for some $f_{1}\in C^{\infty}(M_{1})$ and $f_{2}\in C^{\infty}%
(M_{2})$. This feature of our unperturbed system comes from the singular
dependence of the symplectic form $\sigma$ on the perturbation parameter at
$\varepsilon=0$.

The vector field $\mathbb{V}$ is $\pi_{1}$-related with $v_{f}$ and hence the
trajectories of $\mathbb{V}$ are projected onto trajectories of the
Hamiltonian vector field $v_{f}$ , $\pi_{1}\circ\operatorname{Fl}%
_{\mathbb{V}}^{t}=\varphi^{t}\circ\pi_{_{1}}$. Here, $\varphi^{t}$ denotes the
flow of $v_{f}$. Therefore, $\operatorname{Fl}_{\mathbb{V}}^{t}$ is the
skew-product flow,
\[
\operatorname{Fl}_{\mathbb{V}}^{t}(m_{1},m_{2})=(\varphi^{t}(m_{1}%
),\mathcal{G}_{m_{1}}^{t}(m_{2})),
\]
where $\mathcal{G}_{m_{1}}^{t}$ is a smooth family of symplectomorphisms on
$(M_{2},\sigma_{2})$ determining as the solution of the time-dependent
Hamiltonian system
\[
\frac{d\mathcal{G}_{m_{1}}^{t}(m_{2})}{dt}=V_{F}^{(2)}(\varphi^{t}%
(m_{1}),\mathcal{G}_{m_{1}}^{t}(m_{2})),
\]
\[
\mathcal{G}_{m_{1}}^{0}=\operatorname{id}_{M_{2}}.
\]
Assume that the flow $\operatorname{Fl}_{\mathbb{V}}^{t}$ is periodic with
frequency function $\omega=\frac{2\pi}{T}$. Then,
\begin{equation}
\varphi^{t+T(m_{1},m_{2})}(m_{1})=\varphi^{t}(m_{1})\label{PET1}%
\end{equation}
for all $m_{1}\in M_{1},m_{2}\in M_{2}$ and $t\in\mathbb{R}$. If $v_{f}\neq0$
on $M_{1}$, then differentiating equality (\ref{PET1}) along $M_{2}$ says
that the period function $T$ is independent of $m_{2}$ and hence
$\omega=\varpi\circ\pi_{1}$, for a certain smooth positive function $\varpi$
on $M_{1}$. Therefore, the Hamiltonian flow $\varphi^{t}$ of $v_{f}$ is also
periodic with frequency function $\varpi$.

\begin{theorem}\label{hamperid}
Let $V_{H_{\varepsilon}}=\mathbb{V}+\varepsilon\mathbb{W}$ be the
Hamiltonian vector field on $(M,\sigma)$, where the flow of the
unperturbed vector field $\mathbb{V}$ (\ref{Vec2}) is periodic with
frequency function $\omega$. Assume that the regular set
$\operatorname{Reg}(v_{f})$ is everywhere dense in $M_{1}$ and the
$\mathbb{S}^{1}$-action associated to the Hamiltonian flow $\varphi^{t}$ of
$v_{f}$ is free on $\operatorname{Reg}(v_{f})$. Then, the perturbation
vector field $\mathbb{W}$ (\ref{Vec3}) satisfies the normalization condition
\begin{equation}
\mathcal{L}_{\langle \mathbb{W\rangle}}\omega=0\text{ on }M\label{PET5}%
\end{equation}
and hence by Proposition \ref{propnorm}, $V_{H_{\varepsilon}}$ admits the
normalization of first order with respect to $\mathbb{V}$.
\end{theorem}
\begin{proof}
It is sufficient to show that (\ref{PET5}) holds on the open domain $\pi
_{1}^{-1}(\operatorname{Reg}(v_{f}))$ which is everywhere dense in $M$. The
period-energy relation for $v_{f}$ says that $d\varpi\wedge df=0$ on
$\operatorname{Reg}(v_{f})$ and hence \
\begin{equation}
d\omega\wedge d(f\circ\pi_{1})=0\label{IDPE}%
\end{equation}
on $\pi_{1}^{-1}(\operatorname{Reg}(v_{f}))$. The hypotheses of the theorem
imply that $d(f\circ\pi_{1})\neq0$. On the other hand, taking into account
(\ref{PSF2}), (\ref{PSF4}), (\ref{Vec3}) and the identity $\hat{v}_{f}=V_{f\circ\pi_{1}}^{(1)}$, we get
\begin{align*}
\mathcal{L}_{\mathbb{V}}F &  =\mathcal{L}_{\hat{v}_{f}}F=\pi_{1}^{\ast}%
\sigma_{1}(V_{f\circ\pi_{1}}^{(1)},V_{F}^{(1)})\\
&  =-\pi_{1}^{\ast}\sigma_{1}(V_{F}^{(1)},V_{f\circ\pi_{1}}%
^{(1)})=-\mathcal{L}_{\mathbb{W}}(f\circ\pi_{1}).
\end{align*}
It follows from here that
\begin{align*}
\mathcal{L}_{\langle \mathbb{W}\rangle}(f\circ\pi_{1}) &  =\langle \mathcal{L}_{\mathbb{W}}%
(f\circ\pi_{1})\rangle=-\langle \mathcal{L}_{\mathbb{V}}F\rangle\\
&  =-\mathcal{L}\langle _{\mathbb{V\rangle}}F=0.
\end{align*}
Finally, using these relations and applying the interior product
with $\langle W\rangle$ to both sides of (\ref{IDPE}), we get the
equality
\[
0=(\mathbf{i}_{\langle
W\rangle}d\omega)d(f\circ\pi_{1})-(\mathbf{i}_{\langle
W\rangle}d(f\circ\pi _{1}))d\omega=-(\mathcal{L}_{\langle
W\rangle}\omega)d(f\circ\pi_{1}).
\]
which implies (\ref{PET5}).
\end{proof}

\begin{remark}
In the situation when $v_{f}\equiv0$, we get a Hamiltonian model
which appears in the theory of adiabatic approximation
\cite{ArKN-88}, \cite{Ne-08}. In this case, the periodicity of the
flow of $\mathbb{V}=V_{F}^{(2)}$ does not imply that the
perturbation vector field $\mathbb{W}=V_{F}^{(1)}$
satisfies (\ref{PET5}). The period-energy relation for the
restriction of $V_{F}^{(2)}$ to the symplectic slices
$\{m_{1}\}\times M_{2}$ implies only that $d_{2}\omega\wedge
d_{2}F=0$.
\end{remark}

\textbf{Periodicity conditions}. The periodicity of the flow of
vector field $\mathbb{V}$ (\ref{Vec2}) can be formulated as a
resonance relation. Suppose that the flow $\varphi^{t}$ of $v_{f}$
satisfies all hypotheses of Theorem \ref{hamperid}. Choose a
frequency function $\varpi$ $:M_{1}\rightarrow\mathbb{R}$ of
$\varphi^{t}$ in a such a way that the orbit through every point $m_{1}%
\in\operatorname*{Reg}(v_{f})$ is $\tau(m_{1})$-minimally periodic.
Here $\tau(m_{1})=\frac{2\pi}{\varpi(m_{1})}$. It follows that
$\mathbb{V}$ is complete and the group property of
$\operatorname{Fl}_{\mathbb{V}}^{t}$ implies the relations
\begin{equation}
\mathcal{G}_{m_{1}}^{t_{1}+t_{2}}=\mathcal{G}_{\varphi^{t_{2}}(m_{1})}^{t_{1}%
}\circ\mathcal{G}_{m_{1}}^{t_{2}}\label{Vec4}%
\end{equation}
for any $m_{1}\in M_{1}$ and $m_{2}\in M_{2}$. Define the
\textit{monodromy} of the flow $\operatorname{Fl}_{\mathbb{V}}^{t}$
over a point $m_{1}\in M_{1}$
as a symplectomorphism $g_{m_{1}}:M_{2}\rightarrow M_{2}$ given by $g_{m_{1}%
}:=\mathcal{G}_{m_{1}}^{\tau(m_{1})}$. Then, property (\ref{Vec4})
implies the identity
$\mathcal{G}_{m_{1}}^{t+\tau(m_{1})}=\mathcal{G}_{m_{1}}^{t}\circ
g_{m_{1}}$. It follows that the flow \
$\operatorname{Fl}_{\mathbb{V}}^{t}$ is periodic if and only if
there exists an integer $k\geq1$ such that
\begin{equation}
g_{m_{1}}^{k}=\operatorname*{id}\text{ }\forall m_{1}\in
M_{1}.\label{PET2}
\end{equation}
In this case, the corresponding frequency function can be defined as
$\omega=\frac{1}{k}\varpi\circ\pi_{1}$.

An important class of perturbed dynamics on slow-fast spaces comes
from the linearization procedure for Hamiltonian systems around
invariant symplectic submanifolds \cite{KarVor-98}. In this case,
the unperturbed term $\mathbb{V}$ is a linear vector field on a
symplectic vector bundle which represents the normal linearized
dynamics. The verification of  condition (\ref{PET2}) is related to
computing the monodromy of a time-periodic linear Hamiltonian
system. For several Hamiltonian models with two degree of freedom,
this problem was studied in \cite{Yosh-84}, \cite{Yosh--01}.

\begin{example}
On the phase space $(M=R_{0}^{2}\times R^{2},$ $\sigma=dp_{1}\wedge
dq_{1}+\varepsilon dp_{2}\wedge dq_{2})$, consider the following perturbed
Hamiltonian
\begin{equation}
H_{\varepsilon}=\frac{1}{2}p_{1}^{2}+\frac{1}{4}q_{1}^{4}+\frac{\varepsilon
}{2}(p_{2}^{2}+\frac{\delta}{2}q_{1}^{2}q_{2}^{2})\label{Mod1}%
\end{equation}
where $\varepsilon\ll1$ is the perturbation parameter and $\delta\in
\mathbb{R}$. The corresponding Hamiltonian vector field has the
representation (\ref{Vec1}), where the unperturbed and perturbed
parts are of the form
\begin{equation}
\mathbb{V}=p_{1}\frac{\partial}{\partial q_{1}}-q_{1}^{3}\frac{\partial
}{\partial p_{1}}+p_{2}\frac{\partial}{\partial q_{2}}-\delta q_{1}^{2}%
q_{2}\frac{\partial}{\partial p_{2}},\label{Mod2}%
\end{equation}%
\begin{equation}
\mathbb{W}=-\delta q_{1}q_{2}^{2}\frac{\partial}{\partial p_{1}}.\label{Mod3}%
\end{equation}
Using the results in \cite{Yosh-84}, one can show that the flow of the
vector field $\mathbb{V}$ is periodic (condition (\ref{PET2}) ) if and only if
the parameter $\delta$ satisfies the relation
\[
\sqrt{2}\cos\left(  \frac{\pi}{4}\sqrt{1+8\delta}\right)  =\cos\left(
\pi\frac{n}{k}\right)
\]
for arbitrary coprime integers $n,k\in\mathbb{Z}$ such that
$0 < n < k$. Solutions of this equation are representated
in the form $\delta=\frac{r(r+1)}{2}$, where $r$ runs over the
subset $\displaystyle
\Delta=\mathbb{Q}\cap\bigcup_{s=0}^{\infty}(4s,4s+1).
$ The frequency function is given by $\omega=\frac{1}{c}(2p_{1}^{2}+q_{1}%
^{4})^{\frac{1}{4}}$, where $c=\frac{\sqrt{2}}{\pi}\int_{0}^{1}\frac
{dz}{\sqrt{1-z^{4}}}$. Therefore, for every $r\in\Delta$, vector
fields (\ref{Mod2}), (\ref{Mod3}) satisfy the hypotheses of Theorem
\ref{hamperid}  and hence Hamiltonian system (\ref{Mod1}) admits a
normalization of first order. Notice that this system is
non-integrable for all $\varepsilon \neq0$ and $r\in(4s,4s+1),$
$s\in\mathbb{Z}_{+}$ \cite{Koz-96,Yosh--01}.
\end{example}

\textbf{Hamiltonian Structures}. As we have mentioned above the
vector field $\mathbb{V}$ (\ref{Vec2}) does not inherit any natural
Hamiltonian structure from $V_{H_\varepsilon}$. Here, we formulate a
criterion for the existence of Hamiltonian structure for
$\mathbb{V}$ which is based on Corollary \ref{corolker} and the
results in \cite{DaVo-08,Vor-05}.

Let $\Omega_{\operatorname*{hor}}^{1}(M)$ be the subspace of horizontal
1-forms $\eta$ on $M$ with respect to the projection $\pi_{1}:M\rightarrow
M_{1}$. A 1-form $\eta$ belongs to $\Omega_{\operatorname*{hor}}^{1}(M)$ if
$\mathbf{i}_{Z}\eta=0$ for every vector field $Z$ on $M$ such that $d\pi
_{1}\circ Z=0$. In particular, $d_{1}F\in\Omega_{\operatorname*{hor}}%
^{1}(M)$. Since the infinitesimal generator $\Upsilon=\frac{1}{\omega
}\mathbb{V}$ of the $\mathbb{S}^{1}$-action is $\pi_{1}$-related with
$\frac{2\pi}{\tau}v_{f}$ , the pull-back by the flow $\operatorname{Fl}%
_{\mathbb{V}}^{t}$ leaves invariant the subspace $\Omega_{\operatorname*{hor}%
}^{1}(M)$. It follows that $\Omega_{\operatorname*{hor}}^{1}(M)$ is also an
invariant subspace for operators $\mathcal{A},$ $\mathcal{L}_{\Upsilon}$ and
$\mathcal{S}$.

\begin{theorem}\label{flperiod}
Suppose that the flow of the vector field
$\mathbb{V}=\hat{v}_{f}+V_{F}^{(2)}$ is periodic with frequency
function $\omega$ and there exists a $\mathbb{S}^{1}$-invariant
horizontal 1-form $\mu\in\Omega _{\operatorname*{hor}}^{1}(M)$ such
that
\begin{equation}
\langle d_{1}F\rangle=(\mathbf{i}_{\hat{v}_{f}}\mu)d_{1}\omega. \label{ASP}%
\end{equation}
Then, for a given open domain $N\subset M$ with compact closure and small
enough $\varepsilon\neq0$, the vector field $\mathbb{V}$ is Hamiltonian
relative to the symplectic structure
\begin{equation}
\tilde{\sigma}=\pi_{1}^{\ast}\sigma_{1}+\varepsilon\pi_{2}^{\ast}\sigma
_{2}-\varepsilon d\theta, \label{PET6}%
\end{equation}
and the function
\begin{equation}
\tilde{H}_{\varepsilon}=f\circ\pi+\varepsilon(F-\mathbf{i}_{\hat{v}_{f}}%
\theta). \label{PET7}%
\end{equation}
Here $\theta\in\Omega_{\operatorname*{hor}}^{1}(M)$ is a horizontal 1-form on
$M$ given by
\begin{equation}
\theta=\frac{1}{\omega}\mathcal{S}(d_{1}F)-\frac{1}{\omega^{3}}\mathcal{S}%
^{2}(\mathcal{L}_{\hat{v}_{f}}F)d_{1}\omega+\mu. \label{PET8}%
\end{equation}
Moreover, the functions $f\circ\pi$ and $F-\mathbf{i}_{\hat{v}_{f}}\theta$ are
Poisson commuting first integrals of $\mathbb{V}$.
\end{theorem}
\begin{proof}
According to \cite{DaVo-08,Vor-05}, the vector field $\mathbb{V}$ in
(\ref{Vec2}) is Hamiltonian relative to 2-form (\ref{PET6}) and
function (\ref{PET7}) if the horizontal 1-form
$\theta\in\Omega_{\operatorname*{hor}}^{1}(M)$ satisfies the
homological equation
\begin{equation}
\mathcal{L}_{\mathbb{V}}\theta=d_{1}F.\label{PET9}%
\end{equation}
Indeed, taking into account that $\mathbf{i}_{V_{F}^{(2)}}\theta=0$,
we rewrite equation (\ref{PET9}) in the form
\begin{equation}
\mathcal{L}_{\hat{v}_{f}}\theta+\mathbf{i}_{V_{F}^{(2)}}d\theta
=d_{1}F.\label{PET10}%
\end{equation}
On other hand, we have
\[
\mathbf{i}_{\hat{v}_{f}}\tilde{\sigma}=-d(f\circ\pi_{1})+\varepsilon
(d\mathbf{i}_{\hat{v}_{f}}\theta-\mathcal{L}_{\hat{v}_{f}}\theta)
\]
and $\mathbf{i}_{V_{F}^{(2)}}\tilde{\sigma}=-\varepsilon
(d_{2}F+\mathbf{i}_{V_{F}^{(2)}}d\theta)$. It follows form these
relations and (\ref{PET10}) that
\begin{align*}
\mathbf{i}_{\mathbb{V}}\tilde{\sigma} &  =-d(f\circ\pi_{1}) +\varepsilon(d\mathbf{i}_{\hat{v}_{f}}\theta-\mathcal{L}_{\hat{v}_{f}%
}\theta-d_{2}F-\mathbf{i}_{V_{F\text{ }}^{(2)}}d\theta),\\
&  =-d(f\circ\pi_{1})+\varepsilon(d\mathbf{i}_{\hat{v}_{f}}\theta-d_{1}%
F-d_{2}F).
\end{align*}
Finally, under condition (\ref{ASP}), by Corollary \ref{corolker} a solution $\theta$ to
equation (\ref{PET9}) is given by formula (\ref{PET8}).
\end{proof}

It is clear that if $\langle d_{1}F\rangle=0$, then (\ref{ASP}) is
satisfied for $\mu=0$. But, in general, one can not omit $\mu$ in
condition (\ref{ASP}).

\begin{example}
Consider the phase space $(M=M_{1}\times M_{2},\sigma)$, in the case
when $M_{1}=\mathbb{R}\times \mathbb{S}^{1}$ is the 2-cylinder
equipped with standard symplectic form $\sigma_{1}=ds\wedge
d\varphi$. Let $v_{f}=\varpi(s)\frac{\partial}{\partial s}$ where
$f=f(s)$, and
$\varpi=\frac{\partial f}{\partial s}$. Consider a vector field $\mathbb{V}%
=\varpi(s)\frac{\partial}{\partial s}+V_{F}^{(2)}$, where $F\in
C^{\infty}(M)$ is a function which is independent of $\varphi$ and
$\frac{\partial F}{\partial s}\neq0$. Assuming that
$\frac{\partial\varpi}{\partial s}\neq0$, we get that $\langle
d_{1}F\rangle=d_{1}F\neq0$ and condition (\ref{ASP}) holds for
$\mu=\frac{\partial\varpi}{\partial s}^{-1}\frac{\partial
F}{\partial s}d\varphi$.
\end{example}

Finally, we remark that if $\langle d_{1}F\rangle\wedge
d\omega\neq0$, then condition (\ref{ASP}) does not hold.

\textbf{Invariant symplectic structures}. Consider again the phase space
$(M=M_{1}\times M_{2},\sigma=\pi_{1}^{\ast}\sigma_{1}+\varepsilon\pi
_{2}^{\ast}\sigma_{2})$ and suppose that we are given an $\mathbb{S}^{1}$-action on $M$ ( independent of $\varepsilon$) with infinitesimal generator
of the form $\Upsilon=\hat{v}_{h}+V_{J}^{(2)}$, for some $h\in
C^{\infty}(M_1)$ and $J\in C^{\infty}(M)$. Therefore, $\Upsilon$ is $\pi_{1}%
$-related with a Hamiltonian vector field $v_{h}$ on $(M_{1},\sigma_{1})$.
The $\mathbb{S}^{1}$-action on $M$ descends to a Hamiltonian $\mathbb{S}^{1}%
$-action on $M_{1}$ with infinitesimal generator $v_{h}$. Then, we
have the relations
\begin{equation}
\mathbf{i}_{\Upsilon}\sigma=-d_{1}(h\circ\pi_{1})-d_{2}J, \label{ISS}%
\end{equation}%
\[
\mathcal{L}_{\Upsilon}\sigma=d\mathbf{i}_{\Upsilon}\sigma=-\varepsilon
d_{1}d_{2}J
\]
which say that the symplectic form $\sigma$ is $\mathbb{S}^{1}%
$-invariant only in the case when $J=\pi_{1}^{\ast}j_{1}+\pi_{2}^{\ast}j_{2}$
for some $j_{1}\in C^{\infty}(M_{1})$ and $j_{2}\in C^{\infty}(M_{2})$.

\begin{theorem}
The $\mathbb{S}^{1}$-average of the symplectic form $\sigma$ has the
following representation
\begin{equation}
\langle \sigma\rangle=\sigma-\varepsilon d\beta,\label{SRep1}%
\end{equation}
where $\beta\in\Omega_{\operatorname*{hor}}^{1}(M)$ is the horizontal 1-form
given by
\begin{equation}
\beta=\mathcal{S}(d_{1}J)\equiv\frac{1}{2\pi}\int_{0}^{2\pi}(t-\pi
)(\operatorname{Fl}_{\Upsilon}^{t})^{\ast}d_{1}Jdt.\label{SRep2}%
\end{equation}
For every open domain $N\subset M$ with compact closure and small enough
$\varepsilon$, $\langle \sigma\rangle$ is symplectic form on $N$. The $\mathbb{S}%
^{1}$-action is Hamiltonian relative to $\langle \sigma_{M}\rangle$
if and only if
\begin{equation}
\langle d_{1}J\rangle=0.\label{Adb1}
\end{equation}
Under this condition, the corresponding momentum map is
\begin{equation}
\tilde{J}_{\varepsilon}=h\circ\pi_{1}+\varepsilon(J-\mathbf{i}_{\hat{v}_{h}%
}\beta).\label{Adb2}
\end{equation}
\end{theorem}
\begin{proof}
By formulas (\ref{Clos1}) and (\ref{ISS}), the
$\mathbb{S}^{1}$-average of the symplectic form $\sigma$ is given by
$\langle \sigma\rangle=\sigma-\varepsilon
d\tilde{\beta}$, where%
\begin{align*}
\tilde{\beta} &  =\frac{1}{2\pi}\int_{0}^{2\pi}(t-\pi)(\operatorname{Fl}%
_{\Upsilon}^{t})^{\ast}\mathbf{i}_{\Upsilon}\sigma_{M}dt =-\frac{1}{2\pi}\int_{0}^{2\pi}(t-\pi)(\operatorname{Fl}_{\Upsilon}%
^{t})^{\ast}d_{2}Jdt,\\
&  =-d(\frac{1}{2\pi}\int_{0}^{2\pi}(t-\pi)(\operatorname{Fl}_{\Upsilon}%
^{t})^{\ast}Jdt)+\frac{1}{2\pi}\int_{0}^{2\pi}(t-\pi)(\operatorname{Fl}%
_{\Upsilon}^{t})^{\ast}d_{1}Jdt.
\end{align*}
This implies (\ref{SRep1}). The non-degeneracy of $\langle
\sigma_{M}\rangle$ for small enough $\varepsilon\neq0$, follows from
the evaluating of the 2-form
$\langle \sigma_{M}\rangle$ on the basis of vector fields $\{\hat{v}_{\xi_{i}%
}+V_{\mathbf{i}_{\hat{v}_{\xi_{i}}}\beta}^{(2)},$ $V_{x_{\alpha}%
}^{(2)}\}$, where $\{\xi_{i}\}$ and $\{x_{\alpha}\}$ are coordinates
functions on $M_{1}$ and $M_{2}$, respectively (see, also
\cite{DaVo-08}). Condition (\ref{Adb1}) and formula (\ref{Adb2})
follow directly from Theorem \ref{flperiod}.
\end{proof}

\begin{remark}
If $v_{h}=0$, then we can think of the $\mathbb{S}^{1}$-action as a
family of Hamiltonian actions on $(M_{2},\sigma_{2})$ with
parameterized momentum map $J_{m_{1}}(m_{2})=$ $J(m_{1},m_{2})$. In
this case, hypothesis (\ref{Adb1}), called the adiabatic condition,
was introduced in \cite{MaMoRa-90,Mon-88} in the context of the
Hannay-Berry connections.
\end{remark}

\begin{example}
Consider the phase space $(M=M_{1}\times M_{2},\sigma=\pi_{1}^{\ast}\sigma
_{1}+\varepsilon\pi_{2}^{\ast}\sigma_{2})$, where $M_{2}=\mathbb{S}^{2}
\subset\mathbb{R}^{3}=\{\mathbf{x}=(x^{1},x^{2},x^{3})\}$ is the
unit sphere equipped with standard symplectic form
$\sigma_{2}=\frac{1}{2}\epsilon
_{ijk}x^{i}dx^{j}\wedge dx^{k}$. Given a smooth mapping $\mathbf{\phi}%
:M_{1}\rightarrow\mathbb{S}^{2}$, we define the function $J\in
C^{\infty}(M)$ by
$J(m_{1},\mathbf{x})=\mathbf{\phi}(m_{1})\cdot\mathbf{x}$ for
$m_{1}\in M_{1}$ and $\mathbf{x}\in\mathbb{S}^{2}$. Consider the
$\mathbb{S}^{1}$-action on $M$ with infinitesimal generator
$V_{J\text{ }}^{(2)}=-(\mathbf{\phi \times
x})\cdot\frac{\partial}{\partial x}$, which is given by the
rotations in $\mathbb{R}_{\mathbf{x}}^{3}$ about the axis
$\mathbf{\phi}$. Then, one can show that the
$\mathbb{S}^{1}$-invariant symplectic form $\langle \sigma\rangle$
has the representation (\ref{SRep1}), where
$\beta=(\mathbf{\phi\times x})\cdot d_{1}\mathbf{\phi}$. Moreover,
it easy to see that (\ref{Adb1}) holds, $\langle
d_{1}J\rangle=(\mathbf{\phi\cdot
x})(\mathbf{\phi\cdot}d_{1}\mathbf{\phi})=0$, and hence the
$\mathbb{S}^{1}$-action is Hamiltonian relative to $\langle
\sigma\rangle$ with momentum map $\varepsilon J$. Such a kind of
invariant symplectic structures appears in the study of the particle
dynamics with spin in the context of the averaging method
\cite{Vor-11}.
\end{example}

\textbf{ACKNOWLEDGMENTS.} We thank M.V.Karasev for helpful discussions.This
research was partially supported by CONACYT under the grant no.55463.

\end{document}